\documentclass[10pt,A4,conference]{IEEEtran}
\usepackage{geometry}
 \usepackage{tikz}  
\usetikzlibrary{calc,intersections}
 \usepackage{verbatim}
\usepackage{textcomp}
\usepackage[T1]{fontenc}
\usepackage{bbm}
\usepackage{multirow}
\usepackage{dsfont}
\usepackage{cite}
\usepackage{epsfig}
\usepackage{float}
\usepackage{enumerate}
\usepackage{balance}
\usepackage{color}
\usepackage{caption}
\usepackage{algorithm}
\usepackage[noend]{algpseudocode}
\algrenewcommand\Return{\State \algorithmicreturn{} } 
\usepackage{hyperref}
\hypersetup{
    colorlinks=true,
    linkcolor=blue,
    filecolor=magenta,      
    urlcolor=cyan,
		citecolor=blue,
}

\usepackage{amssymb}
\usepackage{amsthm}
\usepackage{array}
\usepackage{graphicx}
\usepackage{lettrine}
\usepackage{epsf} 
\usepackage{psfrag}
\usepackage[usenames,dvipsnames]{pstricks}
\usepackage{pst-grad} 
\usepackage{pst-plot} 
\usepackage[space]{grffile} 
\usepackage{etoolbox} 

\usepackage{multirow}
\newtheorem*{remark}{Remark}
\newtheorem{theorem1}{Theorem}
\newtheorem{proposition}[theorem1]{Proposition}
\newtheorem{theorem2}{Theorem}
\newtheorem{definition}[theorem2]{Definition}
\newtheorem{theorem3}{Theorem}
\newtheorem{corollary}[theorem3]{Corollary}
\newtheorem{theorem4}{Theorem}
\newtheorem{lemma}[theorem4]{Lemma}
\newtheorem{mydef}{Definition}
\usepackage{optidef}

\begin{document}
\title{The Optimal and the Greedy: Drone Association and Positioning Schemes for Internet of UAVs}

\author{
    \IEEEauthorblockN{Hajar~El~Hammouti,~\IEEEmembership{Member,~IEEE,}
        Doha~Hamza,~\IEEEmembership{member,~IEEE,}
				Basem~Shihada,~\IEEEmembership{Senior member,~IEEE,}
        Mohamed-Slim~Alouini,~\IEEEmembership{Fellow,~IEEE,}
        and Jeff~S. Shamma~\IEEEmembership{Fellow,~IEEE}}
    \IEEEauthorblockA{ CEMSE division, King Abdullah University of Science and Technology
(KAUST), Thuwal, Makkah Province, KSA.\\
\{hajar.hammouti,doha.hamzamohamed,basem.shihada, slim.alouini,jeff.shamma\}@kaust.edu.sa
}
}
\maketitle
\begin{abstract}
This work considers the deployment of unmanned aerial vehicles (UAVs) over a pre-defined area to serve a number of ground users. Due to the heterogeneous nature of the network, the UAVs may cause severe interference to the transmissions of each other. Hence, a judicious design of the user-UAV association and UAV locations is desired. A potential game is defined where the players are the UAVs. The potential function is the total sum rate of the users. The agents' utility in the potential games is their marginal contribution to the global welfare or their so-called wonderful life utility. A game-theoretic learning algorithm, binary log-linear learning (BLLL), is then applied to the problem. Given the potential game structure, a consequence of our utility design, the stochastically stable states using BLLL are guaranteed to be the potential maximizers. Hence, we optimally solve the user-UAV association and 3D-location problem. Next, we exploit the submodular features of the sum rate function for a given configuration of UAVs to design an efficient greedy algorithm. Despite the simplicity of the greedy algorithm, it comes with a performance guarantee of $1-1/e$ of the optimal solution. To further reduce the number of iterations, we propose another heuristic greedy algorithm that provides very good results. Our simulations show that, in practice, the proposed greedy approaches achieve significant performance in a few number of iterations.  

\end{abstract}


\section{Introduction}

According to a recent report of the federal aviation authority (FAA)\cite{FAA}, the number of drones in USA has reached 2 millions in 2019 and is estimated to attain 2.5 millions by 2025.
Indeed, in the near future, thousands of drones are expected to navigate autonomously over cities to deliver a plethora of services such as traffic reporting, package delivery, and public surveillance~\cite{Survey2016,li2018uav}. The main virtue of such technology is the high mobility of drones, their versatile nature, their rapid deployment, and the extremely wide range of services they can provide. 


One of the earliest applications of drone-related services is in the telecommunications industry~\cite{hayat2016survey,fotouhi2019survey}. Equipped with smart transceivers, drones can be deployed as flying base stations that extend coverage in crowded places and remote areas. They can also be deployed as aerial relays that collect or disseminate data in an Internet of things environment. Also, thanks to their fast deployment, drones can be used in a post-disaster scenario to replace damaged ground base stations.

Although the application of drones in the telecommunications industry is very appealing, their efficient deployment still faces several technical challenges that range from trajectory planning to channel modeling~\cite{yan2019comprehensive} and 3D placement~\cite{mozaffari2018tutorial,gupta2016survey,mkiramweni2019survey,bulut2018trajectory,zeng2017energy,ChannelKhuwaja,ravi2016downlink,hayajneh2016optimal}. In this paper, we are interested in the optimal deployment of drones coupled with the optimal drone-users association. Although this problem has been widely investigated in the literature, none of the existing works has provided an optimal solution to the studied problem, specifically when interference is considered. Furthermore, only a few works measure the efficiency of their proposed approach against the optimal one. In this paper, we undertake this task by answering the following questions: What is the approach to guarantee an optimal deployment of UAVs and an optimal drone-users association to maximize the downlink sum-rate, in the presence of interference and with bandwidth and quality of service constraints? What is the cost of such an optimal solution? Are there any alternative approaches that reach an efficient solution to the problem in a fewer number of iterations? How efficient is this solution, and how far is it from optimal? 

\subsection{Related Work}
The optimal 3D placement of UAVs has received considerable attention in the last few years. One of the first works to study the placement of the drones in the 3D space for communications purposes is the work by Mozaffari \textit{et al.} in~\cite{mozaffari2015drone}. In that paper, the authors provide closed-form expressions for an optimal height that maximizes the drones' coverage area. The work mainly focuses on the cases of single and two drones. For the two drones scenario, the authors show that the presence of interference increases the complexity of the system leading to a challenging optimization problem. This problem has been extended in~\cite{mozaffari2016efficient} to a multiple drones scenario. In~\cite{mozaffari2016efficient}, the authors consider interference coming from the nearest neighbor only. This approximation results in a tractable coverage optimization that is solved using circle packing theory. In general, when interference is not considered, the objective function becomes convex with respect to the 3D placement. For this reason, the authors in~\cite{shakhatreh2017efficient} adopt a gradient descent based algorithm to efficiently place the UAV in order to minimize the transmit power required to cover indoor users. The problem of 3D placement to maximize the number of covered users has also been tackled in~\cite{alzenad20173} and has been solved by decoupling horizontal and vertical placements without any loss of optimality.

Moreover, the overall UAV-enabled network performance is tightly related to the number of served users. In the classical network association, users are either served by the closest base station (Voronoi association), or they are assigned to the base station with the best signal-to-interference-and-noise-ratio (weighted Voronoi association). In either case, the distance-only based association may result in highly congested base stations and unbalanced resource allocation across the network. Hence, a handful of works can be found in the literature that study the association rule along with the 3D placement of the UAVs. The joint 3D optimization and user association is challenging to tackle. On one hand, it involves non-convex objective or/and constraint functions with mixed variables (continuous for UAVs locations and binary for the association). On the other hand, even by fixing one of the variables while dealing with the others, the problem remains non-convex and generally NP-hard. To deal with such an issue, one commonly used approach is to decompose the studied optimization into subproblems where each subproblem is addressed separately. The results of each subproblem are used as inputs for the next one, and generally, the process is repeated until convergence is reached. While such an approach can provide satisfactory results, it is not guaranteed to reach the global optimum. When using a decomposition process, the algorithm will often halt at a suboptimal solution with no guaranteed bounds on the suboptimality gap. Furthermore, most of the proposed approaches have no provable convergence properties. 

 For example, by using k-means and particle swarm optimization sequentially, the joint user association and 3D locations was addressed in~\cite{kalantari2017user} in order to maximize the logarithmic rate of the users under delay and backhaul constraints. A similar decomposition approach was proposed in~\cite{farooq2018multi} to first connect devices to the UAVs using matching theory and then optimally place the UAVs in the 2D space using control theory. Using transport theory, cell partitioning was proposed in~\cite{mozaffari2017mobile} to cluster the users, and then, the non-convex 3D placement optimization problem is solved using sequential quadratic programming. In line with the previous works, the approach proposed in~\cite{el2019learn} relies on combining distributed algorithms in order to address the user association, the 2D placement, and the altitude adjustment subproblems separately. Due to the complexity of the studied problem, none of the previously cited papers has provided an approach that exactly solves the target optimization. Indeed, the studied problem is not only non-convex and challenging to solve but it is also NP-hard, meaning that a polynomial-time algorithm that exactly solves the optimization problem does not exist~\cite{woeginger2003exact}. This implies that the optimal solution will necessarily lead to an exponential-time search.
 
 It is important to note that under the terrestrial communications setup, similar resource allocation problems have been investigated, and approaches to reach the exact optimum were proposed. For example, in~\cite{hou2011distributed}, the authors propose an algorithm based on a Gibbs sampler to optimize the joint channel selection and users association in WLAN networks. A more general work was presented in~\cite{borst2013nonconcave} where the authors develop a framework based on Markov Random Fields and Gibbs measures to exactly solve the resource allocation problem in OFDMA networks. Unlike the previously cited works, we tackle the 3D placement problem which is inherent to air-to-ground communications and present a learning mechanism that requires little knowledge of the search space. The learning algorithm, binary log-linear learning (BLLL), is a game-theoretic algorithm that was introduced in \cite{marden2012revisiting} and since then has found wide applicability in wireless communication \cite{xu2016distributed, dai2018energy, zhu2013distributed}. The idea is simple: By designing the agents' utilities, we formulate our problem as a potential game. Then only one agent, a UAV, is active at a time. The active agent compares the utilities of two actions: Its current action and another feasible one. A Gibbs sampler then chooses the actual action based on probabilities calculated from the potential utilities of the two actions. The work in \cite{marden2012revisiting} confirms that such a simple learning rule is guaranteed to linger at the potential maximizers in potential games. 
 
 Since the considered problem is NP-hard, the convergence of BLLL is exponentially slow. Hence, we also provide a greedy algorithm with a performance guarantee of achieving at least one half of the optimal solution. Our greedy approach leverages the submodular properties of the studied problem in order to guarantee an efficient performance. We also refer to the papers~\cite{su2016submodular,lakiotakis2019joint,wang2016low}, that reformulate the resource allocation optimization as a submodular maximization to provide a lower bound approximation on the proposed solutions. These papers either ignore interference in the objective function  (for the rate maximization in~\cite{su2016submodular,wang2016low}) or consider a very specific objective function with innate monotonicity and submodularity properties (for the caching problem in~\cite{lakiotakis2019joint}).

\subsection{Contribution}

Many works exist on the 3D deployment of UAVs and users association for sum-rate maximization. To the best of our knowledge, none provide the exact solution for the studied optimization problem, specifically under the presence of interference between the UAVs. This is especially true since the formulated optimization problem is NP-hard. Hence, our contributions can be summarized as follows.
\begin{itemize}

\item First, we propose a game-theoretic learning algorithm, binary log-linear learning (BLLL), to help the agents, the UAVs, reach an efficient outcome. By formulating the UAV association and location problem as a potential game, BLLL finds the optimal 3D deployment and users association.  
 
\item Second,  we propose a greedy algorithm that exploits  submodular features of the optimization problem when considered for a given configuration of UAVs. The greedy approach provides a performance guarantee of $1-1/e$ of the maximum sum-rate value.

\item Third, we provide a heuristic approach that achieves efficient results in only a few iterations.

\item Finally, we compare the proposed approaches in terms of convergence rate, computational complexity, memory requirement, and exchanged information.
\end{itemize}

\subsection{Organization of the Paper}

The rest of the paper is organized as follows. Next, we describe the adopted system model. Then, we formulate our optimization problem in section~\ref{Prob}. In section.~\ref{Maxi}, we formulate the interactions between UAVs and users as a potential game, and implement the BLLL in order to find the optimal 3D placement and users association that maximize the sum-rate function. Next, we study the submodularity of the objective function and the matroid structure of the constraints in section~\ref{Submodularity}. Two greedy approaches are studied in section~ 
\ref{Greedy}. The proposed algorithms are compared in section~\ref{BLLLVs}. Finally, simulation results are provided in section~\ref{Simul}.

\subsection{Notations}
Throughout the paper, we adopt the following notations. The Cartesian products of two sets $A$ and $B$ is denoted $A\times B$. $|A|$ denotes the cardinality of the set $A$. Vectors and matrices are denoted using boldface letters $\textbf{x}$, whereas scalars are denoted by $x$. 

\section{System Model}\label{Sys}

\begin{figure}
\begin{center}
\includegraphics[scale=0.25]{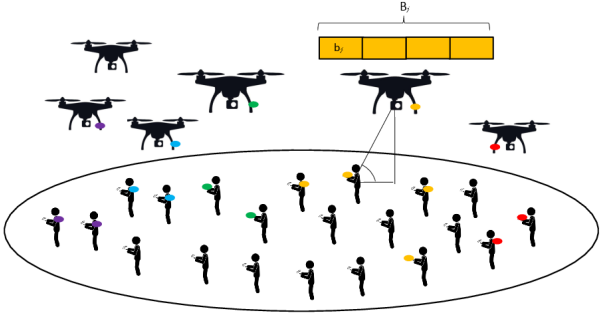}
\end{center}
\caption{Figure shows the setup of our problem. UAVs have designated bandwidths $B_j$ divided into subchannels $b_j$ that they can assign to different users. Users assigned to UAVs are indicated by circles of the same color as the UAV. }
\label{fig:model}
\end{figure}

We assume a drones-enabled network where a set $\cal J$ of $J$ UAVs are deployed over a target area to serve a set $\cal I$ of $I$ ground users. In order to capture the channel variations between the user and the UAV, we adopt the commonly used air-to-ground channel model where the path loss is averaged over line-of-sight (LoS) and non-line-of-sight (NLoS) links where  the probability of LoS is given by ~\cite{al2014modeling}

\begin{equation}
p^{\rm LoS}_{ij}(r_{ij},d_{ij})=\frac{1}{1+\epsilon \cdot{\rm exp}\left(-\beta\frac{180}{\pi}{\rm arctan}\frac{\sqrt{r_{ij}^2-d_{ij}^2}}{d_{ij}}-\epsilon\right)},
\end{equation} 
with $d_{ij}$ is the 2D plane distance from the projected position of UAV $j$ to user $i$, $r_{ij}$ is the distance between the UAV and the user, $\epsilon$ and $\beta$ are environment-dependent parameters.

Consequently, the path loss between UAV $j$ and user $i$ can be formulated  as 

\begin{equation}
L_{ij}(\! r_{ij}\!,\! d_{ij}\!)\!=\!\left(\!\frac{4 \pi\! fr_{ij}}{c}\!\right)^{-\alpha}\!\!\!\!\!\!\!\big(\zeta_{{\rm LoS}}p^{\rm LoS}_{ij}\!(\!r_{ij}\!,\! d_{ij}\!)\!+\!\zeta_{{\rm NLoS}}(\!1\!-\!p^{\rm LoS}_{ij}\!(\!r_{ij}\!,\! d_{ij}\!)\!)\big)^{-1}
\end{equation}

\noindent with $f$ the carrier frequency, $c$ the speed of the light, and $\alpha$ the path loss exponent. $\zeta_{\rm LoS}$ and $\zeta_{\rm NLoS}$  are the parameters for LoS and NLoS losses respectively.

Accordingly, the signal-to-noise-and-interference-ratio (SINR) received at user $i$ from UAV $j$ can be written 

\begin{equation}
    \gamma_{ij}=\frac{P_jL_{ij}(r_{ij},d_{ij})}{\sigma^2+\sum\limits_{k\neq j}P_kL_{ik}(r_{ik},d_{ik})},
\end{equation}

\noindent where $P_j$ is the transmit power of UAV $j$.

We consider the downlink communication channel and are interested in the spectral efficiency $\eta_{ij}$ between a UAV $j$ and a user $i$. $\eta_{ij}$ is given by 

\begin{equation}
 \eta_{ij}={\rm log}_2(1+\gamma_{ij})\big. 
\end{equation}

Due to backhaul limitations, we assume that each UAV $j$ has a limited number of users $N_j$ to connect with. Each UAV equally allocates a bandwidth $b_j$ to its served user. Therefore, each ground user receives a throughput $R_{ij}$ that can be formulated as

\begin{equation}
 R_{ij}=b_j \eta_{ij}.
\end{equation}

\section{Problem Formulation}\label{Prob}
We are interested in the downlink sum-rate of the ground users. Our objective is to optimally deploy the UAVs in the 3D space and associate the users in order to maximize the sum-rate function. 
Let $\textbf{q}=(q_{ij})$ be the binary UAVs-users association matrix and  $(\textbf{x},\textbf{y},\textbf{h})$ the UAVs 3D positions. Let $\mathcal{U}$ be the set of users and $\mathcal{K}$ the set of UAVs. The target optimization problem is formulated as follows.
\vspace{-0.3cm}
\begin{maxi!}
{\textbf{q},(\textbf{x},\textbf{y},\textbf{h})}{\sum \limits_{j \in \mathcal{K}}\sum \limits_{i \in \mathcal{U}} q_{ij}R_{ij}\label{objective}}
{\label{GeneralOptimizationOrigin}}{}
\addConstraint{\sum_{i} q_{ij}\leq N_j,\quad\forall j \in \mathcal{J}\label{waw1}}
{}{}
\addConstraint{ \frac{q_{ij}}{\eta_{ij}}\leq \frac{1}{\eta^{\rm min}},\quad\forall (i,j) \in \mathcal{I}\times\mathcal{J}\label{waw11}}
{}{}
\addConstraint{ x^{\rm min}\leq x_j \leq x^{\rm max}\quad \forall j \in \mathcal{J}\label{waw12}}
{}{}
\addConstraint{ y^{\rm min}\leq y_j \leq y^{\rm max} \quad\forall j \in \mathcal{J}\label{waw13}}
{}{}
\addConstraint{ h^{\rm min}\leq h_j\leq h^{\rm max}  \quad\forall j \in \mathcal{J}\label{waw14}}
{}{}
\addConstraint{\sum_{j} q_{ij}\leq 1,\quad\forall i \in \mathcal{I}\label{wat}}
{}{}
\addConstraint{q_{ij}\in \{0,1\},\quad\forall (i,j) \in \mathcal{I}\times\mathcal{J}\label{wawf3}}
{}{}
\end{maxi!}
Constraint~(\ref{waw1}) ensures that the number of associated users for each UAV $j$ does not exceed its maximum quota of users $N_j$. Constraint~(\ref{waw11}) guarantees a certain quality of service for each associated user by ensuring that its spectral efficiency is no less than a predefined threshold $\eta^{\rm min}$. Constraints~(\ref{waw12}),~(\ref{waw13}) and ~(\ref{waw14}) ensure that the 3D coordinates of all the UAVs belong to a target cubic space. Finally, constraints~(\ref{wat}) and~(\ref{wawf3}) restrict the ground user to be associated with, at most, one UAV.

The problem under analysis is mathematically challenging as it involves a non-convex objective function and a non-convex constraint(constraint~(\ref{waw11})). It also includes integer and continuous variables which makes a mixed interger non-linear programming (MINLP). Moreover, the association problem can be formulated as a knapsack problem, known as NP-hard.  

In the following, we solve the target optimization using a game-theoretic approach. The optimal solution of the studied problem is therefore obtained using binary log-linear learning (BLLL), a learning algorithm that provides guarantees on reaching the maximizers of the objective function when the game is proved to be a potential game.

\section{Sum-Rate Maximization}\label{Maxi}

In order to solve the underlying optimization, we discretize the 3D space and represent it in the form of a 3D grid. We formulate the interactions between UAVs and users as a potential game where the downlink sum-rate is the potential function. Then, BLLL is implemented on UAVs in order to find the optimal 3D placement and users association that maximize the sum-rate function.

\subsection{Game Formulation}
\subsubsection{Background}
In game theory, a potential game is a game where any unilateral change in a player's utility results in an equal change in a global welfare function called \textit{potential function}. Therefore, whenever a player performs an action that improves its utility, he also improves the potential function. More formally, the definition of a potential game is given bellow. 
\begin{definition}[Potential game]~\cite{monderer1996potential}
 Let $\mathcal{A}$ be a set of strategy profiles of a game $\mathcal{G}$. 
$\mathcal{G}$ is a potential game if there exists a potential function $F: \mathcal{A}\longrightarrow \mathbb{R}$ such that for each player ${j}$, $\forall (a_j,\textbf{a}_{-j}) \text{ and } (a_j',\textbf{a}_{-j})~\in~\mathcal{A}$ 
\begin{equation}
F(\! a_j,\textbf{a}_{-j}\!)\!-\!F(a_j'\!,\textbf{a}_{-j}\!)\!=\!U_{j}(a_j\!,\textbf{a}_{-j}\!)\!-\!U_j(a_j'\!,\textbf{a}_{-j}\!),
\end{equation}
\noindent where $U_j$ is the utility of player $j$.
\end{definition}

\subsubsection{UAVs potential game}
Let us consider the 3D grid, where
$X=\{x^{\rm min},x^{\rm min}+\delta x, x^{\rm min}+2 \delta x, \dots, x^{\rm max}\}$, 
$Y=\{y^{\rm min},y^{\rm min}+\delta y, y^{\rm min}+2 \delta y, \dots, y^{\rm max}\}$ and
$H=\{h^{\rm min},h^{\rm min}+\delta h, h^{\rm min}+2 \delta h, \dots, h^{\rm max}\}$
represent x-,y- and z-axis respectively with $\delta x, \delta y, \delta h >0$ are their respective step granularity. Let $Q=\{0,1\}$ be an indicator.  $\mathcal{G}^d=\{\mathcal{J},\mathcal{A},\{U_j\}_1^J\}$ is the game where the UAVs are the players, $\mathcal{A}=X\times Y \times H \times Q^I$ is the set of their actions, and $U_j: \mathcal{A}\rightarrow \mathbb{R}$ the utility function such as, given the 3D deployment of all UAVs and association for all the users, the outcome of UAV $j$ is given by its marginal contribution

\begin{equation}\label{marginal}
U_j(\textbf{x},\textbf{y},\textbf{h},\textbf{q})=\sum \limits_{k \in \mathcal{K}}\sum \limits_{i \in \mathcal{U}} q_{ik}R_{ik}-\sum \limits_{k \in \mathcal{K}\backslash \{j\}}\sum \limits_{i \in \mathcal{U}} q_{ik}R_{ik}(-j), 
\end{equation}

\noindent where $R_{ik}(-j)=b_k {\rm log}_2\big(1+\frac{P_kL_{ik}(r_{ik},d_{ik})}{\sigma^2+\sum\limits_{l\neq j,k}P_lL_{il}(r_{il},d_{il})}\big)$ is the perceived rate at user $i$ when interference from the UAV $j$ is not considered.

\begin{proposition}[UAVs potential game]
The game $\mathcal{G}^d=\{\mathcal{J},\mathcal{A},\{U_j\}_{j=1}^{J}\}$ is a potential game where the potential function is the sum-rate function. 
\end{proposition}
\begin{proof}
This result is straightforward and stems from the design of the utility function.
\end{proof}

\subsection{Binary Log-Linear Learning (BLLL)}
The binary log-linear algorithm is an important algorithm that belongs to the class of log-linear learning algorithms. It has the property of relaxing the synchronous updates of the players' strategies. It also does not require the knowledge of the entire action space nor the utilities of all the actions~\cite{marden2012revisiting}. When the BLLL is implemented, the player selects an action $a$ given an action $a'$ with the probability 
\begin{equation}\label{equaBLLL}
 p(a'\rightarrow a)=\frac{e^{U(a)/T}}{e^{U(a')/T}+^{U(a)/T}},   
\end{equation}
\noindent where $T$ is a tuning parameter called the temperature of the algorithm, and $U$ is the utility of the player.

\begin{corollary}\label{cor}~\cite{marden2012revisiting}
Under BLLL, the only maximizers of the potential function are the stochastically stable states of the algorithm. 
\end{corollary}

The idea behind the BLLL is that each player selects a random strategy (a 3D location and a subset of users to associate with without exceeding its maximum quota) with a probability proportional to $e^{U(a)/T}$, where $T$ is some strictly positive number, also called the temperature. As the temperature goes to $0$, the binary log-linear distribution concentrates on the states that maximize of the potential function. Clearly, the propability to choose an action increases when the utility with respect to this action increases. Hence, the better is the strategy, the higher is the propability to be selected. However, it is far from intuitive how such an updating rule can converge to a global optimum or how it may even converge. The proof of convergence of such a process is based on \textit{the theory of resistance trees} and can be found in details in~\cite{marden2012revisiting}. 

In order to reach the maximum value of the sum-rate, all the UAVs have to adhere to the BLLL process. In Algorithm~\ref{Algo1}, when a UAV wakes up, according to its timer (line~\ref{Line1}), it selects a random location of the 3D grid and a random association with the users, computes the new utility regarding this joint 3D position and association, and then decides whether to move to this new action or not with the probability in equation~(\ref{equaBLLL}). It is important to note that in order to meet constraint~(\ref{waw11}), users who are not satisfied with their spectral efficiency are disconnected from the UAV and their rates are not included in the utility of that UAV. This process is iterated while slowly decreasing the temperature $T$ (line~(\ref{Tem})).


\begin{algorithm}
  \caption{BLLL for joint 3D position and users association selection
    }\label{Algo1}
     \begin{algorithmic}[1] 
			\State{\textbf{Initialization}: }	\State{$(\textbf{x},\textbf{y},\textbf{h})$ random matrix for 3D locations of UAVs}
			\State{$a_{ij}=0,\forall (i,j)\in \mathcal{I}\times \mathcal{J}$}
		     \For{$T\rightarrow 0$}\label{Tem}
		     \For{$j \in \mathcal{J}$}
		     \If{$rand(1)>0.5$}\label{Line1}
		     \State{UAV $j$ selects at random one of these locations $(x_j\pm \delta x, y_j, h_j)$, $(x_j, y_j\pm \delta y, h_j)$,$(x_j , y_j, h_j\pm \delta h)$}\label{Line2}
		     \State{UAV $j$ selects at random a number of unconnected users, within its maximum quota and while respecting the minimum QoS threshold, to associate with.}\label{Line21}
		     \State{UAV $j$ computes $U_j$ with respect to the new position and association as in equation~(\ref{marginal})}\label{Line3}
		     \State{Sample its new joint location and association by using the probability in equation~(\ref{equaBLLL})}\label{Line4}
		     \State{Update its current location, associations, and utility}\label{Line6}
		     \EndIf
		     \EndFor
		     \EndFor
	\end{algorithmic}
\end{algorithm}

It is to be noted that in order to reach the global optimum, BLLL requires an exponential time for convergence. To circumvent this problem, we propose a greedy approach that is guaranteed to converge to at least $1-1/e$ the optimum. Our greedy algorithm relies on the submodular property of the objective function that we discuss in the next section.

\section{Submodularity of the Objective Function}\label{Submodularity}

In this section, we proceed to analyze the submodularity of our objective function and the matroid structure of the contsraints similar to the approach in \cite{su2016}. This analysis will facilitate the greedy algorithm which we employ to solve our problem. First, we introduce the mathematical definitions of submodularity and matroids. Then, we reformulate the problem as a set function maximization problem under partition matroid constraints. This lays the foundation for the greedy algorithm we use to solve our problem in the next section.

\subsection{Basic Definitions}
Assume a ground set $V$. Let $2^V$ be the collection of all subsets of $V$. In discrete optimization, we say that a set function $f: 2^V \rightarrow \mathbb{R}$ is submodular if it satisfies the following property.

\begin{mydef}[Submodularity\cite{edmonds2003submodular}]\label{def:submod2}
 $f$ is said to be submodular if for $A\subseteq B\subseteq V$ and $a\in V\setminus B$:
 \begin{equation}
    f(A\cup \{a\})-f(A)\geq f(B\cup \{a\})-f(B)
\end{equation}
\end{mydef}
An intuitive interpretation of submodularity suggests that the marginal gain of adding an element $a$ to a small set $A$ ($A$ subset of $B$) is greater or equal to adding the same element to the larger set $B$. 


 
 Furthermore, a set function is said to be monotone if its value increases when more elements are added to a set. More formally,

\begin{definition}[Monotonicity]\label{def:mon}
A set function $f: 2^V \rightarrow \mathbb{R}$ is monotone if $ \forall A\subseteq B \subseteq V $
\begin{equation}
    f(A)\leq f(B).
\end{equation}
\end{definition}

We will shortly show that some of our constraints can be described as matroids. A matroid is an algebraic structure that generalizes the concept of independent vectors in linear algebra. In particular:

\begin{definition}[Matroid]
A matroid $M=(V,\cal L)$ consists of a non-empty finite set ground $V$ and a non-empty collection $\cal L$ of subsets of $V$ that satisfy the following properties: 
\begin{enumerate}
\item $\emptyset \in \mathcal{L}$
    \item If $I \in \cal L$ and $J \subset I$, then $J \in \cal L$.
    \item $I, J \in \cal L$ and $|I|>|J|$, then there exists $e \in I\backslash J$ such that $J \cup\{e\} \in \cal L$
\end{enumerate}
\end{definition}

The first two conditions describe the \textit{''hereditary property''}. This property suggests that each subset of a set in the collection $\mathcal{L}$ inherits the independence property. The third condition is usually called the \textit{''augmentation property''}. It implies that each element of the collection can be augmented to a larger set while maintaining the independence property.

When the ground set $V$ is partitioned into disjoint subsets $V_1,V_2,\dots,V_t$, where $t$ is a strictly positive integer, a particular class of matroids, called \textit{partition matroid} emerges:
 
\begin{definition}[Partition matroid]
A partition matroid $M=(V,\cal L)$ is a matroid such that $V$ is partitioned into $t$ disjoint partition sets $V_1,V_2,\dots,V_t$ and ${\cal L}= \{X\subseteq V:|X\cap V_i|\leq k_i,\,\, \forall i=1,2,\dots t\}$, where $0 \leq k_i \leq |V_i|$ are some given integer parameters.
\end{definition}

Now that we have introduced the proposed mathematical framework, let us reformulate the studied problem as a set function maximization problem. 

\subsection{Problem Reformulation}

To begin, let $\mathcal{K}$ be the set of all possible configurations of the UAVs\footnote{A network configuration designates a given network realization where the 3D locations of UAVs are fixed at some positions of the 3D grid.}. Since there are $J$ UAVs and $L$ locations, there are $K=L^{2J-1}$ possible configurations where $K=|\mathcal{K}|$. We also add a null UAV for each user to allow for the possibility that some users will be unassigned. 

We define our ground set $V=\{(i,j,k): i\in{\cal I},j\in{\cal J},k\in{\cal K}\}$. $V$ contains all the tuples formed by users, UAVs, and network configurations. We then partition the ground set $V$ into $K$ disjoint subsets, $V_1^C,V_2^C,...,V_K^C$, where $V_k^C=\{(i,j,k), i\in \mathcal{I}, j\in \mathcal{J},\}$ where $V_k^C$ is the set of all possible associations under a given configuration $k$ and the superscript indicates that the partition is according to the configuration index. Hence, the constraint that only one configuration is possible can be written as finding set $A\in I_C$ where
\begin{equation}
 I_C=
\begin{cases}
   A\subseteq V:|A\cap V_k^C|\leq e_k \mbox{ for some k}\in{\cal K},\\
   A\subseteq V:|A\cap V_n^C|=0 \,\,\,\forall\,\,\, n\in{\cal K}\setminus \{k\}
\end{cases}
\end{equation}
\noindent where $e_k$ is some number denoting the intersection of the two sets\footnote{$e_k$ can be seen as the total number of associated users for a given configuration.}. This constraint merely implies that, in the end, only one cofiguration is selected. 

\begin{remark}
It is noted that we could also set up a constraint for the UAV quota and another for the users' quota. However, we will show that this is not needed. We simply delegate the UAV quota, the users' quota and the minimum spectral efficiency conditions, conditions (\ref{waw11}), (\ref{wat}) and (\ref{wawf3}) in the optimization problem, to the set function evaluation. We also note that considering the set function evaluation over configurations helps us fix the interference experienced by users for a given configuration. As we show in the proof, this helps recover monotonicity and submodularity of the set function evaluated over a given configuration.
\end{remark}

\begin{proposition}
$M_C=(V,I_C)$ is a partition matroid.
\end{proposition}

\begin{proof}
We consider feasible sets $A\subseteq B\subset {\mathcal V}$. To maintain feasibility, $A$ and $B$ must belong to the same configuration. The proof follows immediately using the approach in \cite{matroids}.
\end{proof}

In light of the above definitions, our optimization problem can now be written as:

\begin{maxi!}
{A\in {2^V}}{f(A)\label{objective2}}
{\label{GeneralOptimizationOrigin2}}{}
\addConstraint{A\in I_C\label{wawz}}
{}{}
\end{maxi!}

The above problem can be equivalently written as:

\begin{equation}\label{ProbRefor}
\underset{k\in {\mathcal K}}{\text{ maximize}}\underset{A\in {I_C}}{\text{ maximize  }}{f^k(A)}
\end{equation}

\noindent where,

\begin{equation}
    f^k(A)=\sum\limits_{v_{ijk}\in A} R^k_{ij}.
\end{equation}

\noindent and $f^k(.)$ refers to the function evaluation over a given configuration. Since we must enforce that $A\in {I}_{C}$, we can only consider sets taking elements that belong to the same configuration, and $R^k_{ij}$ is the rate of user $i$ when it is associated with UAV $j$ and configuration $k$ is adopted. We now use the superscript $k$ to emphasize that the rate of user $i$ with UAV $j$ is calculated for a particular configuration $k$, so that we can set the interference for a particular configuration at a constant value.

\begin{proposition}
$f^k(.)$ is monotone and submodular.
\end{proposition} 

\begin{proof}
We prove monotonicity first. Without loss of generality (WLOG), consider two subsets $A\subseteq B\subseteq V_k^C$, i.e., belonging to the same configuration set $k$. 
Let $A$ contain 4 UAVs with a given association for the users. Let $B$ contain $A$ in addition to another UAV with its associated users, then $f^k(A)\leq f^k(B)$ {is always true.} 

We proceed to prove submodularity. Consider any subset $A\subseteq B\subseteq V_k^C$ and $a\in V_k^C\setminus B$. WLOG, let $A$ be the set containing possible associations for users with UAV $j$ at configuration $k$ such that $|A|=N_j-1$, let $B=A\cup \{b\}$, where $b$ is some feasible element to be added to the set of users associated with UAV $j$. It is clear that $|B|=N_j$, hence $B$ is at UAV$_j$'s quota limit. WLOG, let $\{b\}=\arg\min \limits_{B} R^k_{ij}$, i.e. $\{b\}$ is also the element with the minimum contribution to the value $f^k(B)$. Now, consider the addition of another feasible element, $a$ to sets $A$ and $B$:

$$f^k(A\cup \{a\})-f^k(A)=f^k(\{a\}),$$

\noindent while
\begin{equation}
f^k\!(\!B\cup \{a\}\!)\!-\!f^k\!(\!B\!)\!=\!\begin{cases}   
0\!<\!f^k(\!\{a\}\!) \mbox{ if } \eta_{b}\geq\eta_{a},\\
f^k\!(\!\{a\}\!)\!-\!f^k\!(\!\{b\}\!)\!<\!f^k\!(\!\{a\}\!), \mbox{ if } \eta_{a}\!>\!\eta_{b}.\\
\end{cases}
\end{equation}
\noindent where in the above, and with a slight abuse of notation, we use $\eta_b$ to denote the spectral efficiency of element $b$. In evaluating $f^k(B)$, and since $B$ is already at its quota limit, we compare the spectral efficiencies of the existing element $b$ with $a$. 
\end{proof}

\subsection{$K$ Instances of the Greedy Algorithm}

Using the fact that $f^k(.)$ is monotone and submodular, we can now use a simple greedy algorithm to find the locations and associations for the UAVs and users. We use a greedy algorithm to evaluate the maximum for $f^k(.)$ for a given configuration, and then exhaustively find the maximum value for the set function over all configurations. The overall guaranteed performance is $1-1/e$-optimal. This is facilitated by the following lemma:


\begin{lemma}\label{Lemma}
Let $(\mathcal{P})$ be the problem of maximizing a monotone and submodular set function, i.e. $f^k(.)$. Consider the greedy algorithm which starts with an empty set $A_0$, and at each iteration $i$, it adds an element $e$ that maximizes $f(A_{i-1}\cup \{e\})-f(A_{i-1})$, i.e., 
\begin{equation}
    A_{i}=A_{i-1} \cup \{argmax_{\{e\}} f(A_{i-1}\cup \{e\})-f(A_{i-1}\}).
\end{equation}
The greedy algorithm provides $1-1/e$-approximation to the optimal solution of $(\mathcal{P})$ \cite{nemhauser1978}. 
\end{lemma}




While the above greedy algorithm ensures a good network performance, it requires listing all the possible configurations, which is time and memory consuming. However, we do not in fact need to list all the possible configurations. One approach to reduce the search space is to select the locations that are critical and are most likely to provide the best performance; in particular, the barycenters of the users' concentrations. For this purpose, we first run k-means as described in \textbf{Algorithm}~\ref{kmeans}. Each UAV is moved in the 2D plan to the barycenter of a cluster of users. The users within the same cluster are selected based on their SINR. Specifically, users are grouped with the UAV that maximizes their SINR. Hence, k-means selects the best 2D locations based on an SINR criterion. Then, a list of 3D configurations is formed by the 2D locations and the various possible heights of the setup. This process will drastically reduce the number of possible configurations without joepardizing performance as we show in the numerical simulations. The k-means combined with the greedy algorithm is described in \textbf{Algorithm}~\ref{kmeansCom}.

\begin{algorithm}
  \caption{K-means
    }\label{kmeans}
     \begin{algorithmic}[1] 
			\State{\textbf{Initialization}:}
			\State{ UAVs randomly scattered in the 2D space, }
			\State{ $C_j=\emptyset$, $\forall j \in \mathcal{J}$ }
			\State{Choose $N$, the maximum number of iterations.}
			\For{$n= 1:N$}
		    	\For{$j \in \mathcal{J}$}
		        	\For{$i \in \mathcal{I}$}
		            	\If{ $i=argmax_i \eta_{ij}$}
	            		\State{$C_j=C_j\cup\{i\}$}
		            	\EndIf
		           	\EndFor
		    	\State{$x_j=\frac{\sum \limits_{i \in \mathcal{I}} x_i}{|C_j|}$}
		    	\State{$y_j=\frac{\sum \limits_{i \in \mathcal{I}} y_i}{|C_j|}$}
		        \EndFor
		\EndFor
	\end{algorithmic}
\end{algorithm}

\begin{algorithm}
  \caption{Combined K-means and greedy
    }\label{kmeansCom}
     \begin{algorithmic}[1] 
			\State{\textbf{Initialization}:}
			\State{ Run \textbf{Algorithm}~\ref{kmeans} to reduce the number of 2D points. }
			\State{ List in $\mathcal{K}$ all the possible configurations of UAVs that are formed by the 2D points and the studied heights.}
			\For{$k \in \mathcal{K}$}
			\State{Let $A_0$ be an empty set}
			\For{$i=1:I\times J$}
			\State{$A_{i}=A_{i-1} \cup \{argmax_{\{e\}} f(A_{i-1}\cup \{e\})-f(A_{i-1}\})$, such that $e$ satisfies constraints~(\ref{waw11}), (\ref{wat}) and (\ref{wawf3}) }
			\EndFor
			\EndFor
	\end{algorithmic}
\end{algorithm}

\section{Greedy Approaches}\label{Greedy}
In this section, we describe the greedy algorithm that efficiently solves the underlying optimization with $1-1/e$-approximation. Then, we provide a faster heuristic, with no guaranteed performance, which achieves very good results in practice. 

\subsection{Greedy Algorithm}

\begin{algorithm}
  \caption{Greedy algorithm
    }\label{AlgoGree}
     \begin{algorithmic}[1] 
			\State{\textbf{Initialization}:}
			\State{$S=0$, initilization of the maximal sum-rate}\label{Line412}
			\State{$N_j^{\rm Current}=0, \forall j \in \mathcal{J}$, initialization of number of associated users to each UAV.}
				\State{$e^{\rm best}=e_1$, initilization of the best configuration}\label{Line41}
						\State{$q_{ij}^k=0, \forall i \in \mathcal{I}, j\in \mathcal {J}, k\in \mathcal{K}$ }

			\State{$\mathcal{L}=\mathcal{I}$, the set of not associated users}
			\For{$e_k $ a potential configuration}
			\For{ $n=1:I\times J$}
			\State{find $(i,j)$ s.t. $(i,j)=argmax_{(i,j)}(R_{ij}^{k})$}\label{G1}
			\If{user $\sum_{i}q_{ij}^k=0$ is not associated, $\sum_{ij}q_{ij}^k\leq N_{j}$, and $\eta_{ij}\geq \eta^{\rm min}$}
			\State{$q_{ij}^k=1$}\label{G2}
			\State{$N_j^{\rm Current}=N_j^{\rm Current}+1$}\label{G3}
			\EndIf
			\EndFor
			\If{$\sum_{(i,j)}q_{ij}^kR_{ij}^{k}> S$}\label{G4}
			\State{$S=\sum_{(i,j)}qa_{ij}^kR_{ij}^{k}$}
			\State{$e^{\rm best}=e_k$}\label{G5}
			\State{$\textbf{q}^{\rm best}=\textbf{q}^k$}
			\EndIf
			\EndFor
	\end{algorithmic}
\end{algorithm}
As stated in \textbf{Lemma}~\ref{Lemma}, the greedy algorithm will start by selecting the maximum rate for each configuration. Indeed, as described in \textbf{Algorithm}~\ref{AlgoGree}, for each configuration, the greedy algorithm connects the user and the UAV associated with the maximum rate among all pairs users-UAVs of the selected configuration (line~(\ref{G1})). The associated user is then removed from the list of considered users and the quota of its serving UAV is decremented (lines~(\ref{G2}),(\ref{G3})). Then, the second best rate is considered, and the associated user-UAV pair are connected. This process is repeated until all users are either associated or cannot be provided with satisfying rates (i.e. constraint (\ref{waw11}) cannot be satisfied for unassociated users), or all UAVs reach their maximum quota. At each configuration, the algorithm compares with the previous configurations (line~(\ref{G4})). If the selected configuration provides better sum-rate, than the best configuration is updated (line~(\ref{G5})). The process is repeated until all configurations are tested.

\subsection{Adapted Version of the Greedy}

\begin{algorithm}
  \caption{Adapted greedy algorithm
    }\label{Algo4}
     \begin{algorithmic}[1] 
			\State{\textbf{Initialization}:}
			\State{Sort the UAVs in a decreasing order according to their maximum quota, let $\hat{\mathcal{J}}$ be the set of ordered UAVs}\label{Line413}
						\State{$q_{ij}=0, \forall i \in \mathcal{I}, j\in \mathcal{J}$}

			\State{$\mathcal{L}=\mathcal{I}$, the set of not associated users}
			\For{$j \in \hat{\mathcal{J}}$}
			\State{Find the 3D location that maximizes the sum of the best $N_j$ users'rates, where the users belong to the set $\mathcal{L}$}\label{Line42}
			\State{Update the location of UAV $j$}
			\State{Associate UAV $j$ with the $N_j$ best non associated users, from $\mathcal{L}$, for which the quality of service is satisfied  }\label{Line43}
			\State{Update $\mathcal{L}$ by removing associated users}\label{Line44}
			\EndFor
		     
	\end{algorithmic}
\end{algorithm}
At this stage of the paper, we are rather interested in developing a fast algorithm that does not come with a guaranteed performance but provides very good results in practice. We refer to this algorithm as \textit{the adapted version of the greedy}. 

In \text{Algorithm}~\ref{Algo4}, we first sort the UAVs in decreasing order according to their maximum quota (line~(\ref{Line41})). The first UAV selects among all the possible locations of the 3D grid the one that provides the best sum-rate of the best $N_j$ users' rates (where $N_j$ is the maximum quota of UAV $j$) (line~(\ref{Line42})). The $N_j$ users with the best rates are therefore associated with the UAV (line~(\ref{Line43})), and their association is never reconsidered in the next steps of the algorithm (line~(\ref{Line44})). Then, the process is repeated for the remaining UAVs and users. The process ends when the UAV with the minimum quota has been associated with its users.

\section{BLLL vs. Greedy: a Fair Comparison}\label{BLLLVs}

In this section, we compare the previously proposed approaches in terms of convergence rate, computational complexity, memory requirement, and exchanged information. This comparison is summarized in TABLE~\ref{tab:my_label}.
\subsection{Convergence Time}

The BLLL search approach allows us to select an action with a certain probability. This probability is dependent on the utility of the action and the temperature parameter. The higher the utility, the higher the probability that it will be selected. Initially, the temperature is set to a high value in order to allow a wide exploration of the search space. As iterations increase, the temperature is cooled down in order to eliminate unsuccessful strategies. Clearly, the convergence rate of the BLLL depends on two main parameters: the initial temperature, and the cooling scheme of the temperature. It has been shown in \cite{brusco2014comparison} that the logarithmic scheme is one of the most efficient temperature decays. This scheme suggests that at each iteration $t$, the temperature is given by $T(t)=\frac{T_0}{log(1+t)}$, where $T_0$ is the initial temperature. Although such a cooling approach allows a very slow decrease of the temperature, it ensures the convergence to the global optimum when enough iterations are provided. It is also important to note that when the initial temperature is too low, the search space will be reduced, and the algorithm can get trapped in a local optimum. One guideline is to tune the initial temperature based on the first realizations of the utility function, or to set the initial temperature to a high value.

The greedy algorithm also requires a large number of iterations, especially if the search space is not reduced. This is because it has to go through all the possible configurations of the network. However, when we remove configurations that are unlikely to be efficient, the convergence time is significantly reduced. In general, the greedy algorithm will take at most $K\times I\times J$  iterations, where $K$ is the number of possible configurations. On the other hand, the adapted greedy will only require $J$ iterations to converge. 

\subsection{Computational Complexity and Memory Requirement}
From a computational perspective, the UAVs perform simple algebraic operations when they adhere to the BLLL. Essentially, the active UAV, as well as the other UAVs, need to observe the impact of the action on the throughput of their users. Then, each UAV has to compute and broadcast its aggregated throughput (i.e., local sum-rate of its served users) to the active UAV. Also, the UAVs have only to memorize the utility of their previous action, leading to very low memory requirements. 

Similarly, the greedy algorithm does not require computational complexity as it only computes the rates of the users at various UAVs locations. However, it requires high memory storage capacity as it compares rates at different heights.

On the other hand, the adapted greedy approach requires more computational efforts as every UAV has to solve a local optimization problem. In particular, the first UAV has to select, among all the possible locations in the 3D gird, the one that maximizes its local sum-rate. Similarly, the second UAV chooses from the remaining locations the one that maximizes its aggregated throughput. This process is repeated for all the UAVs, one by one, who select from the remaining locations the ones that improve their local sum-rates. At the same time, the algorithm does not require significant memory storage.


\subsection{Exchanged Information}

Based on its formulation in equation (\ref{marginal}), the UAV's utility relies on global and complete information of the network when BLLL is adopted. Indeed, in order to fit into the potential game framework, the utility is designed as the marginal contribution of the player (i.e. UAV). This implies that each UAV has to compute the sum-rate of all associated users when this UAV is part of the game and when it is not. Clearly, significant knowledge is required. Not only does the UAV need to know the throughput of its served users at its selected 3D location, but also knowledge is needed of the throughput of users that are connected to all other UAVs. This will entail a considerable amount of exchanged information in the network. Unfortunately, the convergence of the BLLL to the global optimum comes at the expense of complete network knowledge. Instead, less information is needed in interference-limited networks. In such networks, each UAV's utility depends on the information available at its neighboring UAVs only (i.e., UAVs within interference range). This may drastically reduce the amount of exchanged information while still achieving good performance.

Compared to the BLLL, the greedy algorithm implementation is centralized. This also suggests a high information exchange between users, UAVs, and the centralized entity.
Instead, the adapted version of the greedy involves much less information exchange. At each iteration of the algorithm, the UAV needs only to observe the throughput of its served users. No information is required from the previously deployed UAVs.       
\begin{table*}
    \centering
    \begin{tabular}{|c|c|c|c|c|c|}
    \hline
         & Convergence & Computations & Memory& Information exchange & Implementation  \\
          \hline
         BLLL& Exponential & Algebraic & Small& Complete knowledge & Distributed\\
          \hline
         Greedy & $K\times I \times J $ & Algebraic&  Large &Complete knowledge & Centralized \\
          \hline
         Adapted Greedy & $J$ & Solve local optimization & Small & Local knowledge & Distributed \\
          \hline
    \end{tabular}
    \caption{Comparison of proposed algorithms}
    \label{tab:my_label}
\end{table*}
\section{Simulation Results}\label{Simul} 
To assess the performance of the studied algorithms, we consider the following scenario. 
We assume $45$ users, randomly scattered in an area of $1000\times 1000 m^2$. $5$ UAVs are considered to provide connectivity to the ground users. The UAVs positions are initially set to some random locations as shown in Fig. All the drones are assumed to transmit with the same power $P=10$ dBm. In order to account for the path loss, we suppose $\zeta_{\rm LoS}=1$ dB, $\zeta_{\rm NLoS}=20$ dB, $\alpha=9.61$, $\beta=0.16$, $f_c=2$ GHz, and $c=3*10^8 m/s$. The simulation settings are summarized in TABLE~\ref{Tab}.
\begin{table}
\begin{center}
\begin{tabular}{|c|c||c|c|}
\hline
Parameter & Value & Parameter & Value\\
\hline
Area & $1000\times1000$ & $\delta_x$& $10$m\\
\hline

$\delta_y$& $10$m& $\delta_h$& $10$m  \\

\hline 
$h^{\rm min}$ & $100$m & $h^{\rm max}$ & $200$m\\

\hline
$\eta^{\rm min}$& -3 dB & $I$ &$45$ \\
\hline

$J$& $5$ & $P_j$& $10$ dBm\\
\hline

$\alpha$& $9.61$ & $\beta$& $0.16$\\
\hline

$c$&$3.10^8m\!/s$ & $\zeta_{\rm LoS}$& $1$ dB\\
\hline

$\zeta_{\rm NLoS}$ & $20$dB &$N_j$ & 4\\
	\hline	

\end{tabular}
\caption{Simulation settings.}
\label{Tab}
\end{center}
\end{table}

Fig.~\ref{figInitial}(a) plots the initial 3D locations of UAVs for the studied scenario. In Fig.~\ref{figInitial}(b), we plot the 3D grid for all possible positions of the UAVs, while in Fig.~\ref{figInitial}(c), we show the selected 3D positions after the reduction of the search space using k-means.  

 Fig.~\ref{figSum} plots the network sum-rate vs. the number of iterations. As can be seen from Fig.~\ref{figSum}, although BLLL requires the highest number of iterations to converge, it still provides the best performance. On the other side, less performance is achieved when the greedy algorithm is adopted. However, only a few iterations are needed to reach an efficient value of the network sum-rate. Finally, only fewer iterations are needed for the adapted version of the greedy algorithm in order to ensure convergence. The number of iterations for the adapted greedy approach is equal to the number of UAVs. The adapted greedy, however, achieves the lowest performance compared to the greedy and BLLL approaches. It is also to be noted that the results provided by all approaches are above 1-1/e maximum.

 In Fig.~\ref{figMov}, we plot the 3D movement of the UAVs under the studied algorithms setup. It can be seen from Fig.~\ref{figMov}(a), that the UAVs move sequentially in the 3D space before reaching their final 3D locations. Each UAV finds its best location in order to cover the maximum number of users allowed by its quota. The heights and 2D coordinates of UAVs are adjusted in order to reduce interference and ensure the best network sum-rate. In Fig.~\ref{figMov}(b), it can be noticed that each UAV has only to move once in order to reach its final location. This is because the greedy algorithm will not allow the UAVs to move unless a better location is found. In the studied scenario, the best locations for UAVs were found in the second iteration. The adapted version of the greedy, plotted in Fig.~\ref{figMov}(c), allows one UAV movement at a time. The UAVs are moved one by one to the 3D location that maximizes their aggregate sum-rate. 
\begin{figure}
    \centering
    \includegraphics[scale=0.5]{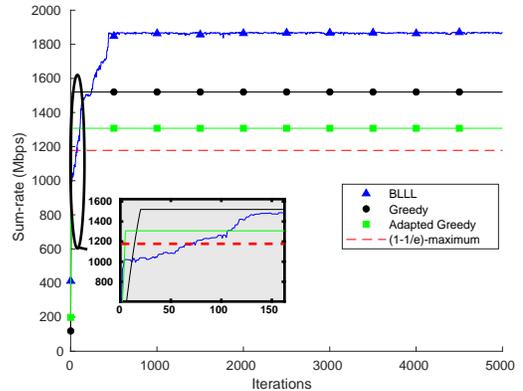}
    \caption{Sum-rate convergence under BLLL, greedy and adapted gready algorithms.}
    \label{figSum}
\end{figure}

\begin{figure*}
\begin{tabular}{ccc}
    \includegraphics[scale=0.38]{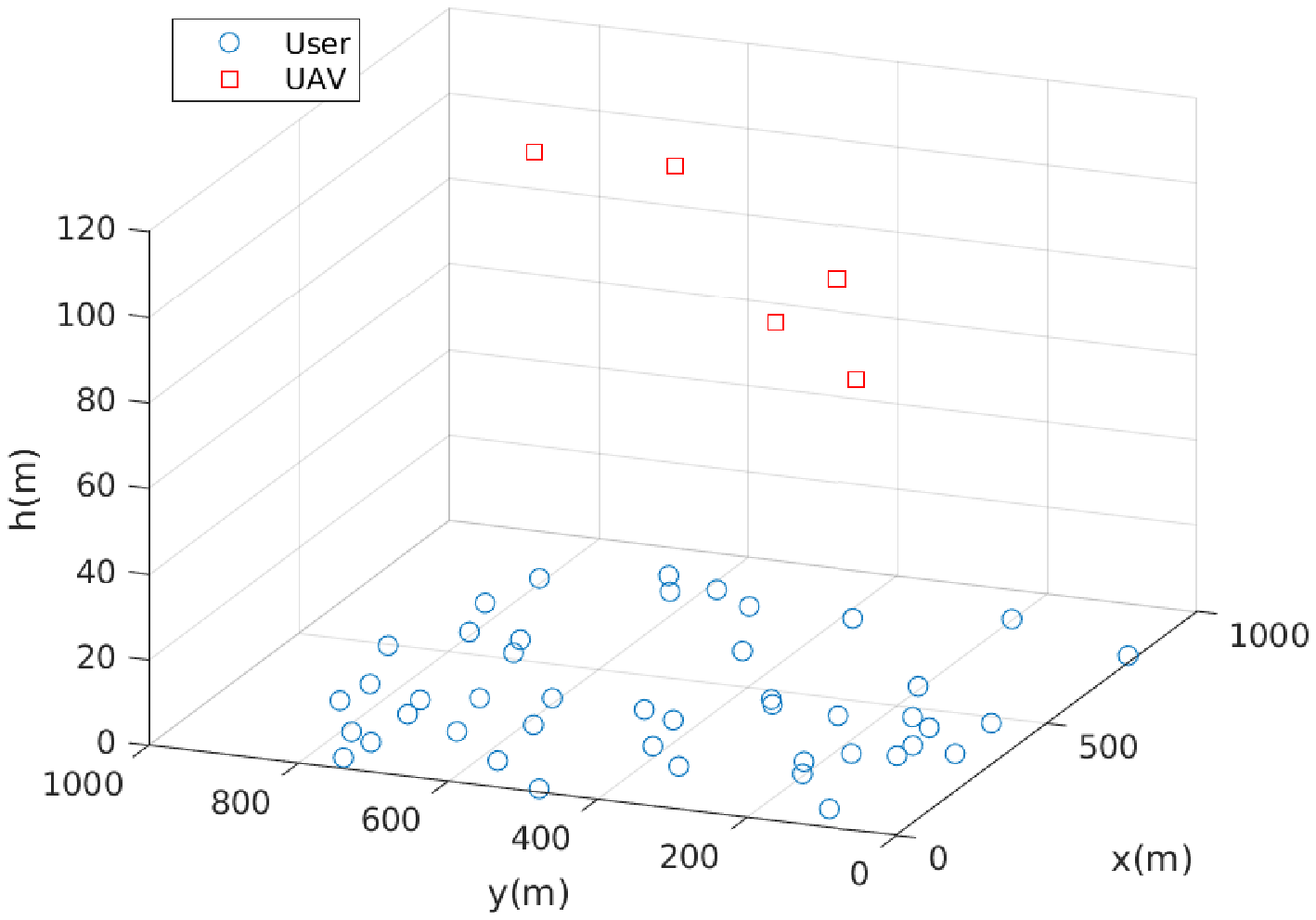}&
    \includegraphics[scale=0.38]{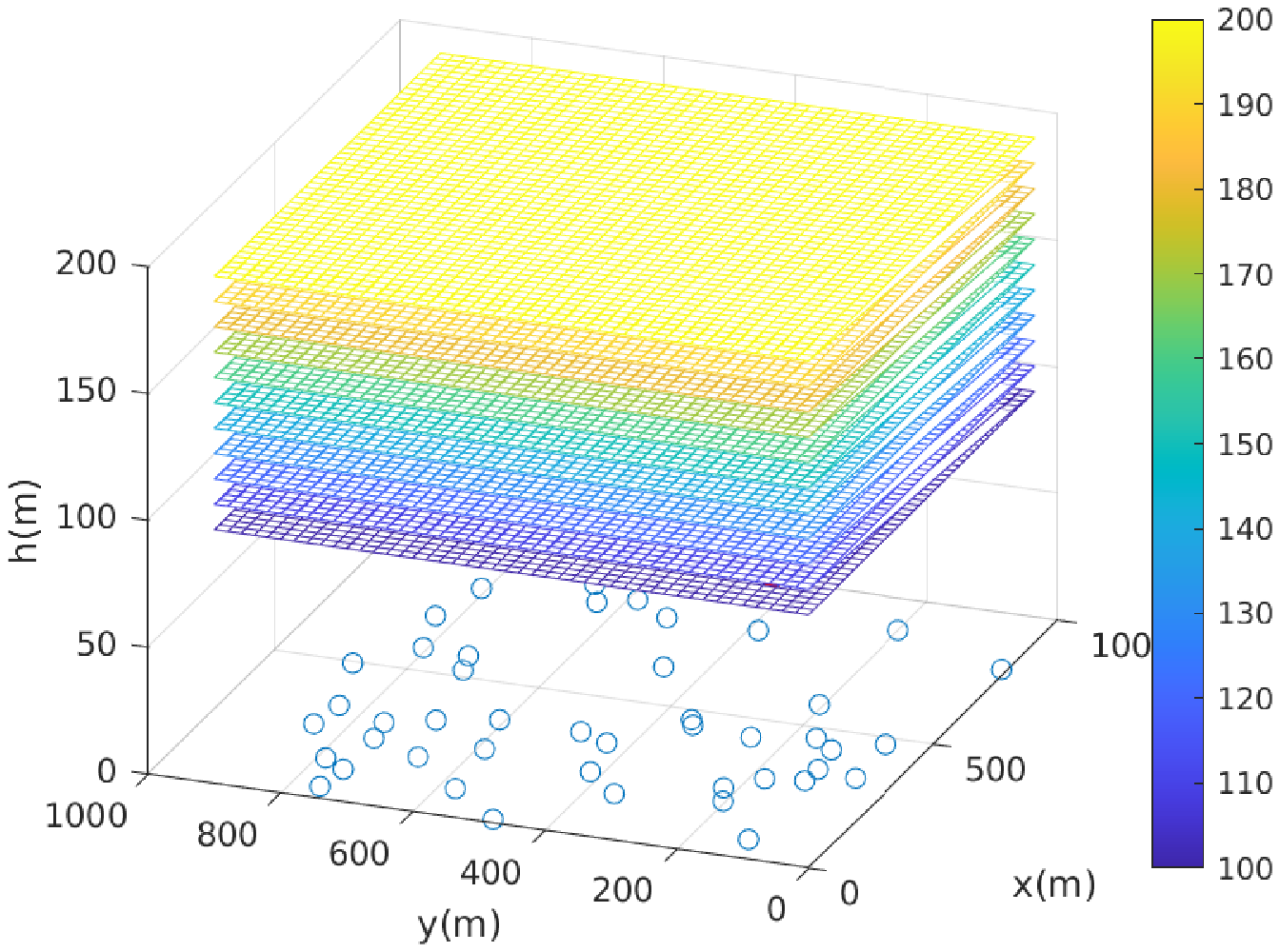}&
    \includegraphics[scale=0.38]{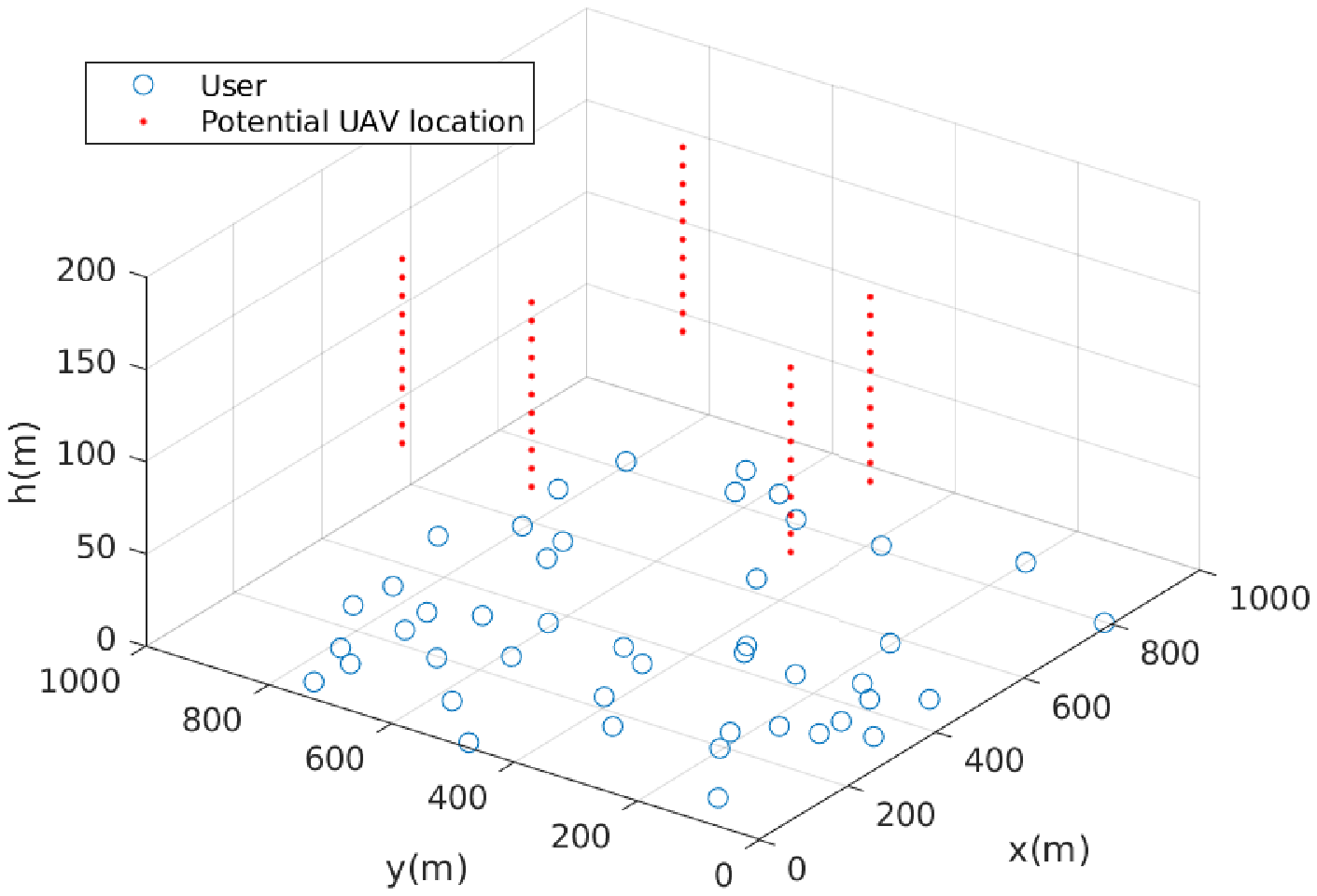}\\
    (a) & (b) & (c)
    \end{tabular}
    \caption{(a) Initial configuration (b) Search space (c) Potential locations of UAVs using k-means.}
    \label{figMov}
\end{figure*}

\begin{figure*}
\begin{tabular}{ccc}
    \includegraphics[scale=0.38]{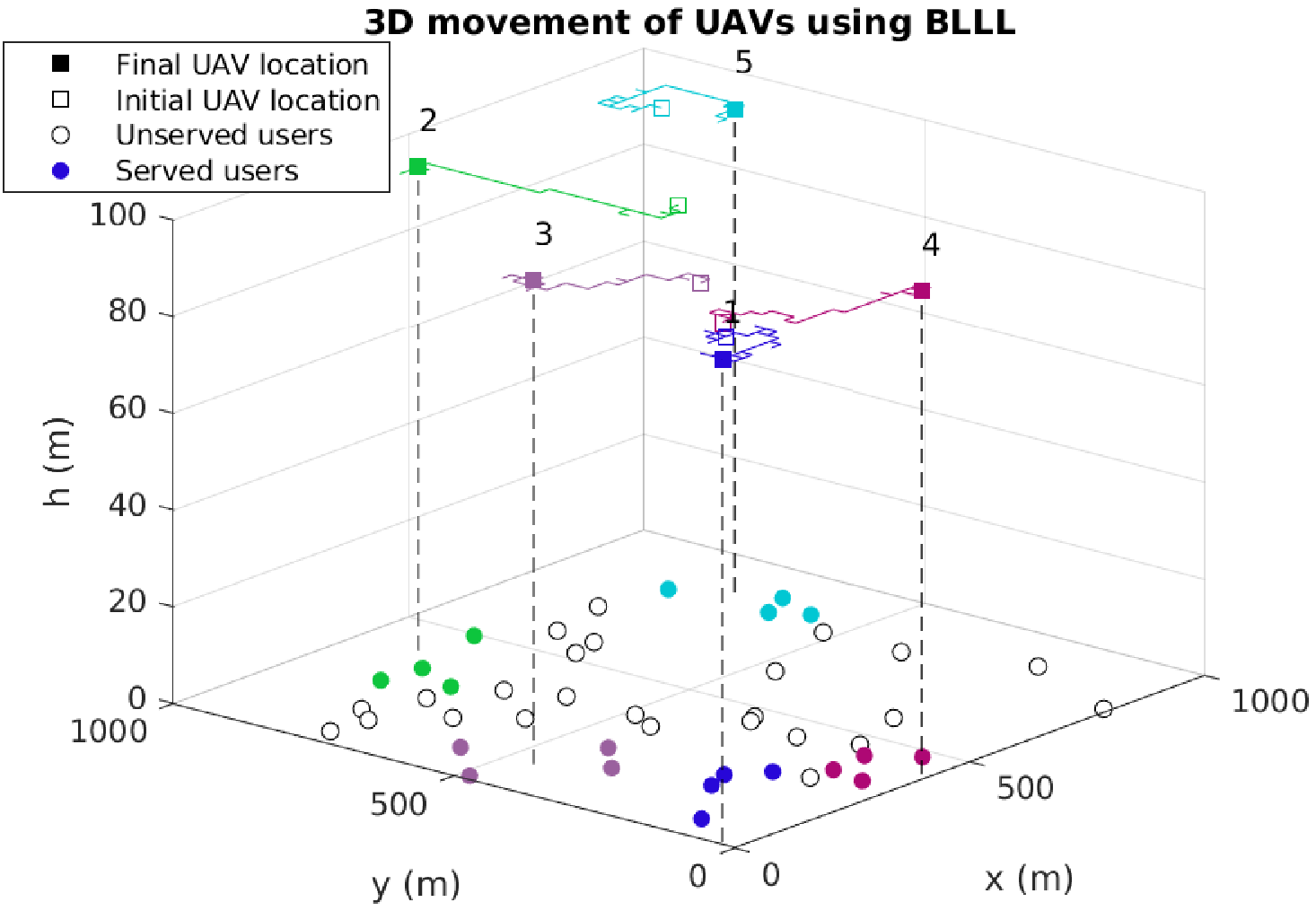}&
    \includegraphics[scale=0.38]{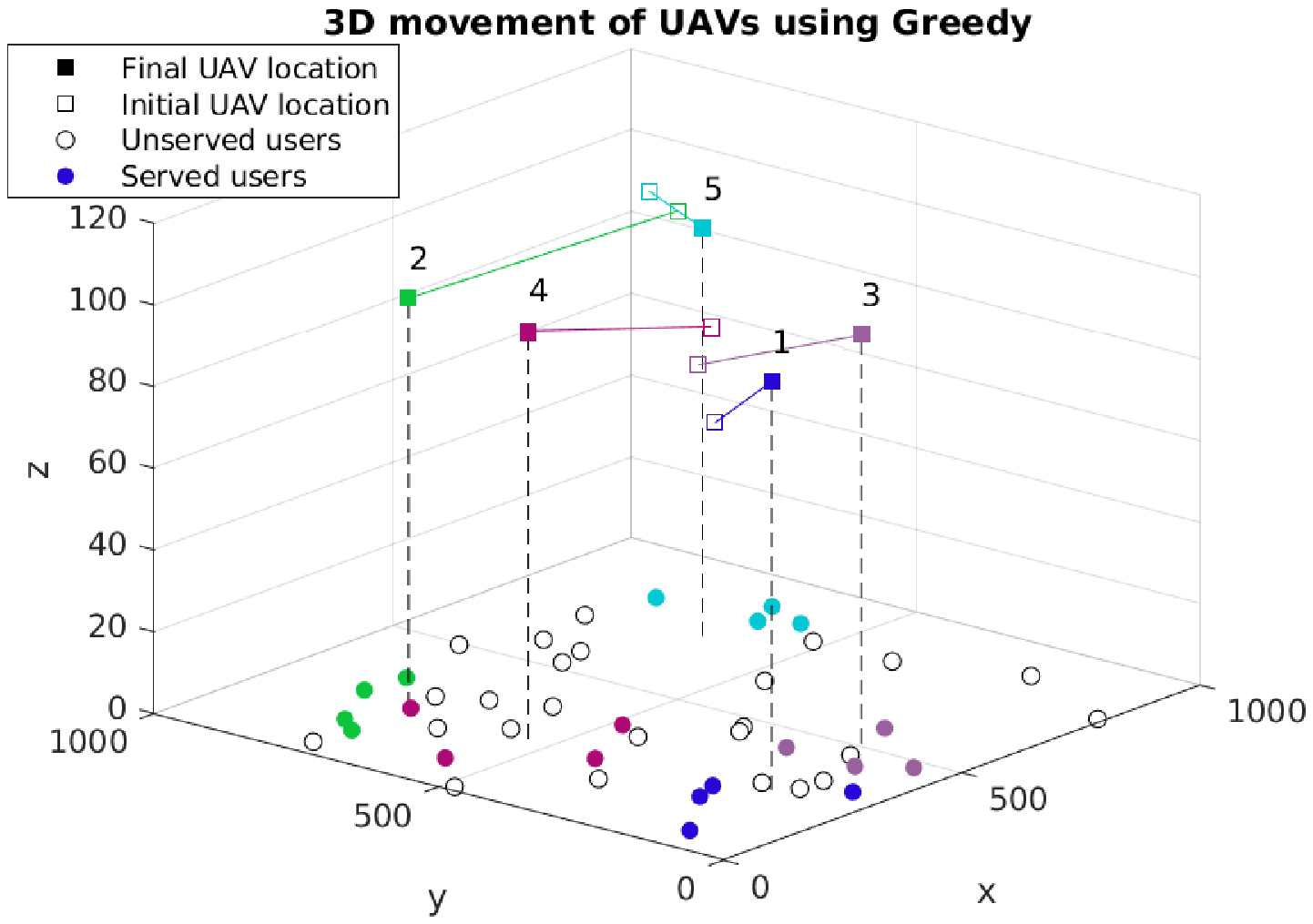}&
    \includegraphics[scale=0.38]{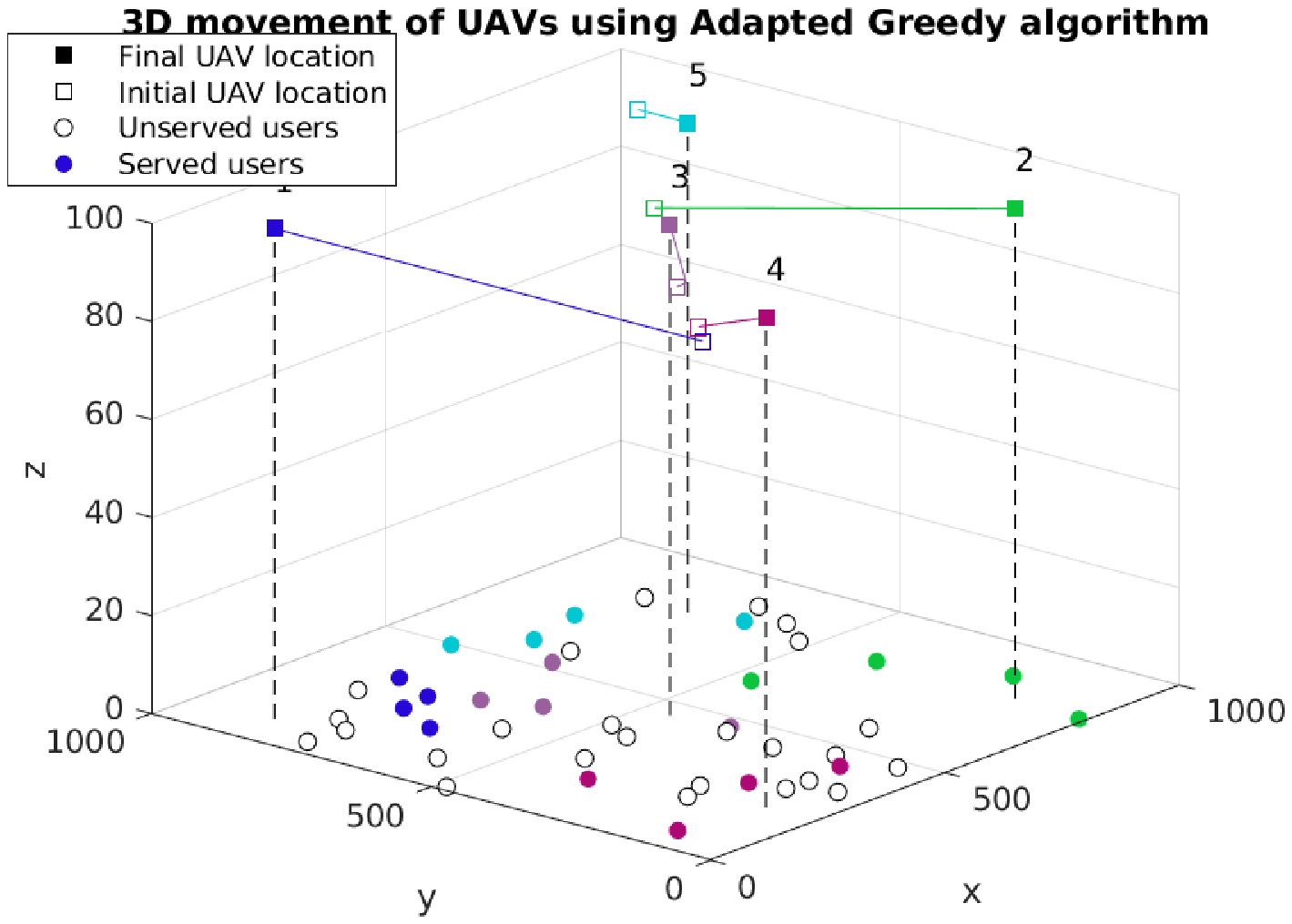}\\
    (a) & (b) & (c)
    \end{tabular}
    \caption{3D movements of UAVs under (a) BLLL, (b) greedy, and (c) adapted greedy algorithms.}
    \label{figInitial}
\end{figure*}
\begin{figure}
    \centering
    \includegraphics[scale=0.6]{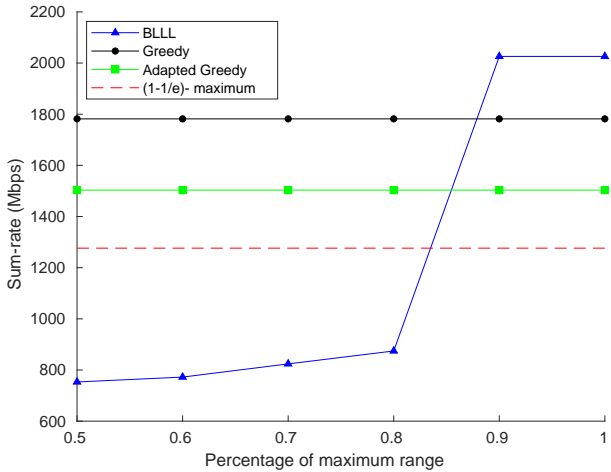}
    \caption{Impact of the neighborhood range on BLLL performance.}
    \label{Range}
\end{figure}
In Fig.~\ref{Range}, we consider a scenario of $10$ UAVs and $60$ users. We plot the final value of the sum-rate (i.e. after convergence of algorithms) vs. the percentage of neighborhood range of UAVs. In this figure, we show that when each UAV receives limted information from the environment (only the UAVs within its range can share the information about their current aggregate rate), a less performance is achieved by the BLLL. As can be seen from the figure, the greedy and adapted greedy algorithms outperform the BLLL when limited information is exchanged (i.e., small range is adopted).  
\section{Conclusion}
In this paper, we addressed the problem of joint 3D placement and users association in UAVs-enabled networks. We proposed three algorithms. The first is guaranteed to reach the global optimum of the sum-rate function at the expense of exponential convergence time. The second exploits the submodularity of the studied problem and has a performance guarantee of $1-1/e$-approximation. The third requires only a fewer iterations, and while it has no guaranteed performance, it achieves very good results in simulations.

\bibliographystyle{IEEEbib}
\bibliography{BiblioUAV}

\begin{thebibliography}{10}

\bibitem{FAA}
www.google.com/loon Federation of aviation~authority (FAA),
\newblock ``Unmanned aircraft systems,''
\newblock 2019.

\bibitem{Survey2016}
N~.H. Motlagh, T.~Taleb, and O.~Arouk,
\newblock ``Low-altitude unmanned aerial vehicles-based internet of things
  services: Comprehensive survey and future perspectives,''
\newblock {\em IEEE Internet of Things Journal}, vol. PP, no. 99, pp. 1--27,
  2016.

\bibitem{li2018uav}
Bin Li, Zesong Fei, and Yan Zhang,
\newblock ``Uav communications for 5g and beyond: Recent advances and future
  trends,''
\newblock {\em IEEE Internet of Things Journal}, vol. 6, no. 2, pp. 2241--2263,
  2018.

\bibitem{hayat2016survey}
S.~Hayat, E.~Yanmaz, and R.~Muzaffar,
\newblock ``Survey on unmanned aerial vehicle networks for civil applications:
  A communications viewpoint.,''
\newblock {\em IEEE Communications Surveys and Tutorials}, vol. 18, no. 4, pp.
  2624--2661, 2016.

\bibitem{fotouhi2019survey}
Azade Fotouhi, Haoran Qiang, Ming Ding, Mahbub Hassan, Lorenzo~Galati Giordano,
  Adrian Garcia-Rodriguez, and Jinhong Yuan,
\newblock ``Survey on uav cellular communications: Practical aspects,
  standardization advancements, regulation, and security challenges,''
\newblock {\em IEEE Communications Surveys \& Tutorials}, vol. 21, no. 4, pp.
  3417--3442, 2019.

\bibitem{yan2019comprehensive}
Chaoxing Yan, Lingang Fu, Jiankang Zhang, and Jingjing Wang,
\newblock ``A comprehensive survey on uav communication channel modeling,''
\newblock {\em IEEE Access}, vol. 7, pp. 107769--107792, 2019.

\bibitem{mozaffari2018tutorial}
M.~Mozaffari, W.~Saad, M.~Bennis, Y.~Nam, and M.~Debbah,
\newblock ``A tutorial on {UAVs} for wireless networks: Applications,
  challenges, and open problems,''
\newblock {\em arXiv preprint arXiv:1803.00680}, 2018.

\bibitem{gupta2016survey}
L.~Gupta, R.~Jain, and G.~Vaszkun,
\newblock ``Survey of important issues in uav communication networks,''
\newblock {\em IEEE Communications Surveys \& Tutorials}, vol. 18, no. 2, pp.
  1123--1152, 2016.

\bibitem{mkiramweni2019survey}
Mbazingwa~Elirehema Mkiramweni, Chungang Yang, Jiandong Li, and Wei Zhang,
\newblock ``A survey of game theory in unmanned aerial vehicles
  communications,''
\newblock {\em IEEE Communications Surveys \& Tutorials}, vol. 21, no. 4, pp.
  3386--3416, 2019.

\bibitem{bulut2018trajectory}
E.~Bulut and I.~G{\"u}ven{\c{c}}l,
\newblock ``Trajectory optimization for cellular-connected {UAVs} with
  disconnectivity constraint,''
\newblock in {\em IEEE International Conference on Communications (ICC)},
  Kensas,USA, May 2018.

\bibitem{zeng2017energy}
Y.~Zeng and R.~Zhang,
\newblock ``Energy-efficient {UAV} communication with trajectory
  optimization,''
\newblock {\em IEEE Transactions on Wireless Communications}, vol. 16, no. 6,
  pp. 3747--3760, 2017.

\bibitem{ChannelKhuwaja}
A~A. Khuwaja, Y.~Chen, N.~Zhao, M.-S. Alouini, and P.~Dobbins,
\newblock ``A survey of channel modeling for {UAV} communications,''
\newblock {\em IEEE Communications Surveys Tutorials}, vol. 20, no. 4, pp.
  2804--2821, 2018.

\bibitem{ravi2016downlink}
V.~V.~C. Ravi and H.~Dhillon,
\newblock ``Downlink coverage probability in a finite network of unmanned
  aerial vehicle ({UAV}) base stations,''
\newblock in {\em IEEE International Workshop on Signal Processing Advances in
  Wireless Communications (SPAWC)}, Edinburgh, United Kingdom, Jul. 2016, pp.
  1--5.

\bibitem{hayajneh2016optimal}
A.~M. Hayajneh, S.~A.~R. Zaidi, D.~C. McLernon, and M.~Ghogho,
\newblock ``Optimal dimensioning and performance analysis of drone-based
  wireless communications,''
\newblock in {\em IEEE Global communications conference workshops (GLOBECOM)},
  Washington DC, USA, Dec. 2016.

\bibitem{mozaffari2015drone}
M.~Mozaffari, W.~Saad, M.~Bennis, and M.~Debbah,
\newblock ``Drone small cells in the clouds: Design, deployment and performance
  analysis,''
\newblock in {\em IEEE Global Communications Conference (GLOBECOM)}, San Diego,
  USA, Dec. 2015, pp. 1--6.

\bibitem{mozaffari2016efficient}
M.~Mozaffari, W.~Saad, M.~Bennis, and M.~Debbah,
\newblock ``Efficient deployment of multiple unmanned aerial vehicles for
  optimal wireless coverage,''
\newblock {\em IEEE Communications Letters}, vol. 20, no. 8, pp. 1647--1650,
  2016.

\bibitem{shakhatreh2017efficient}
Hazim Shakhatreh, Abdallah Khreishah, Ayoub Alsarhan, Issa Khalil, Ahmad
  Sawalmeh, and Noor~Shamsiah Othman,
\newblock ``Efficient 3d placement of a uav using particle swarm
  optimization,''
\newblock in {\em 2017 8th International Conference on Information and
  Communication Systems (ICICS)}. IEEE, 2017, pp. 258--263.

\bibitem{alzenad20173}
M.~Alzenad, A.~El-Keyi, F.~Lagum, and H.~Yanikomeroglu,
\newblock ``{3-D} placement of an unmanned aerial vehicle base station
  {(UAV-BS)} for energy-efficient maximal coverage,''
\newblock {\em IEEE Wireless Communications Letters}, vol. 6, no. 4, pp.
  434--437, 2017.

\bibitem{kalantari2017user}
E.~Kalantari, I.~Bor-Yaliniz, A.~Yongacoglu, and H.~Yanikomeroglu,
\newblock ``User association and bandwidth allocation for terrestrial and
  aerial base stations with backhaul considerations,''
\newblock in {\em International Symposium on Personal, Indoor, and Mobile Radio
  Communications (PIMRC)}, Montreal, QC, Canada, Sep. 2017, pp. 1--6.

\bibitem{farooq2018multi}
M.~J. Farooq and Q.~Zhu,
\newblock ``A multi-layer feedback system approach to resilient connectivity of
  remotely deployed mobile {Internet of Things},''
\newblock {\em IEEE Transactions on Cognitive Communications and Networking},
  vol. 4, no. 2, pp. 422--432, 2018.

\bibitem{mozaffari2017mobile}
M.~Mozaffari, W.~Saad, M.~Bennis, and M.~Debbah,
\newblock ``Mobile unmanned aerial vehicles {(UAVs)} for energy-efficient
  internet of things communications,''
\newblock {\em IEEE Transactions on Wireless Communications}, vol. 16, no. 11,
  pp. 7574--7589, 2017.

\bibitem{el2019learn}
Hajar El~Hammouti, Mustapha Benjillali, Basem Shihada, and Mohamed-Slim
  Alouini,
\newblock ``Learn-as-you-fly: A distributed algorithm for joint 3d placement
  and user association in multi-uavs networks,''
\newblock {\em IEEE Transactions on Wireless Communications}, vol. 18, no. 12,
  pp. 5831--5844, 2019.

\bibitem{woeginger2003exact}
G.~J Woeginge,
\newblock ``Exact algorithms for {NP}-hard problems: A survey,''
\newblock in {\em Combinatorial Optimization-Eureka, You Shrink!}, pp.
  185--207. 2003.

\bibitem{hou2011distributed}
I-Hong Hou and Piyush Gupta,
\newblock ``Distributed resource allocation for proportional fairness in
  multi-band wireless systems,''
\newblock in {\em 2011 IEEE International Symposium on Information Theory
  Proceedings}. IEEE, 2011, pp. 1975--1979.

\bibitem{borst2013nonconcave}
Sem~C Borst, Mihalis~G Markakis, and Iraj Saniee,
\newblock ``Nonconcave utility maximization in locally coupled systems, with
  applications to wireless and wireline networks,''
\newblock {\em IEEE/ACM Transactions on Networking}, vol. 22, no. 2, pp.
  674--687, 2013.

\bibitem{marden2012revisiting}
J.~R Marden and Jeff~S J.~S~Shamma,
\newblock ``Revisiting log-linear learning: Asynchrony, completeness and
  payoff-based implementation,''
\newblock {\em Games and Economic Behavior}, vol. 75, no. 2, pp. 788--808,
  2012.

\bibitem{xu2016distributed}
Yuhua Xu, Jinlong Wang, and Qihui Wu,
\newblock ``Distributed learning of equilibria with incomplete, dynamic, and
  uncertain information in wireless communication networks,''
\newblock in {\em Game Theory Framework Applied to Wireless Communication
  Networks}, pp. 63--86. IGI Global, 2016.

\bibitem{dai2018energy}
Haibo Dai, Yongming Huang, Yuhua Xu, Chunguo Li, Baoyun Wang, and Luxi Yang,
\newblock ``Energy-efficient resource allocation for energy harvesting-based
  device-to-device communication,''
\newblock {\em IEEE Transactions on Vehicular Technology}, vol. 68, no. 1, pp.
  509--524, 2018.

\bibitem{zhu2013distributed}
Minghui Zhu and Sonia Mart{\'\i}nez,
\newblock ``Distributed coverage games for energy-aware mobile sensor
  networks,''
\newblock {\em SIAM Journal on Control and Optimization}, vol. 51, no. 1, pp.
  1--27, 2013.

\bibitem{su2016submodular}
Gongchao Su, Bin Chen, Xiaohui Lin, Hui Wang, and Lemin Li,
\newblock ``A submodular optimization framework for outage-aware cell
  association in heterogeneous cellular networks,''
\newblock {\em Mathematical Problems in Engineering}, vol. 2016, 2016.

\bibitem{lakiotakis2019joint}
Emmanouil Lakiotakis, Pavlos Sermpezis, and Xenofontas Dimitropoulos,
\newblock ``Joint optimization of uav placement and caching under battery
  constraints in uav-aided small-cell networks,''
\newblock in {\em Proceedings of the ACM SIGCOMM 2019 Workshop on Mobile
  AirGround Edge Computing, Systems, Networks, and Applications}, 2019, pp.
  8--14.

\bibitem{wang2016low}
Junyuan Wang, Huiling Zhu, Lin Dai, Nathan~J Gomes, and Jiangzhou Wang,
\newblock ``Low-complexity beam allocation for switched-beam based multiuser
  massive {MIMO} systems,''
\newblock {\em IEEE Transactions on Wireless Communications}, vol. 15, no. 12,
  pp. 8236--8248, 2016.

\bibitem{al2014modeling}
A.~Al-Hourani, S.~Kandeepan, and A.~Jamalipour,
\newblock ``Modeling air-to-ground path loss for low altitude platforms in
  urban environments,''
\newblock in {\em IEEE Global Communications Conference (GLOBECOM)}, Austin,
  USA, Dec. 2014, pp. 2898--2904.

\bibitem{monderer1996potential}
D.~Monderer and L.~Shapley,
\newblock ``Potential games,''
\newblock {\em Games and Economic Behavior}, vol. 14, no. 1, pp. 124--143,
  1996.

\bibitem{su2016}
Chen B. Lin X. Wang H. Li~L Su, G.,
\newblock ``A submodular optimization framework for outage-aware cell
  association in heterogeneous cellular networks,''
\newblock {\em Mathematical Problems in Engineering}, 2016.

\bibitem{edmonds2003submodular}
Jack Edmonds,
\newblock ``Submodular functions, matroids, and certain polyhedra,''
\newblock in {\em Combinatorial Optimization—Eureka, You Shrink!}, pp.
  11--26. Springer, 2003.

\bibitem{matroids}
P.~{Sermpezis}, T.~{Spyropoulos}, L.~{Vigneri}, and T.~{Giannakas},
\newblock ``Femto-caching with soft cache hits: Improving performance with
  related content recommendation,''
\newblock in {\em GLOBECOM 2017 - 2017 IEEE Global Communications Conference},
  Dec 2017, pp. 1--7.

\bibitem{nemhauser1978}
George~L Nemhauser, Laurence~A Wolsey, and Marshall~L Fisher,
\newblock ``An analysis of approximations for maximizing submodular set
  functions—i,''
\newblock {\em Mathematical programming}, vol. 14, no. 1, pp. 265--294, 1978.

\bibitem{brusco2014comparison}
Michael~J Brusco,
\newblock ``A comparison of simulated annealing algorithms for variable
  selection in principal component analysis and discriminant analysis,''
\newblock {\em Computational Statistics \& Data Analysis}, vol. 77, pp. 38--53,
  2014.

\end{thebibliography}

\end{document}